\documentclass[letterpaper, 10 pt, conference]{ieeeconf} 
\usepackage{cite}
\usepackage{amsmath,amssymb,amsfonts}
\usepackage{mathrsfs}
\usepackage{algorithmic}
\usepackage{graphicx}
\usepackage{textcomp}
\usepackage{booktabs}
\usepackage[colorlinks,
            linkcolor=black,
            anchorcolor=blue,
            citecolor=red,
            urlcolor=black
            ]{hyperref}
\newtheorem{theorem}{Theorem}

\IEEEoverridecommandlockouts
\begin{document}
\title{\LARGE \bf
Distributionally Robust Model Predictive Control with Mixture of Gaussian Processes*}
\author{Jingyi Wu and Chao Ning
\thanks{*This work was partially supported by the National Natural Science Foundation of China under Grants 62473256 and 62103264 (Corresponding author: Chao Ning). }
\thanks{J. Wu and C. Ning are with the Department of Automation, Shanghai Jiao Tong University.}
}

\maketitle

\begin{abstract}
Despite the success of Gaussian process based Model Predictive Control (MPC) in robotic control, its applicability scope is greatly hindered by multimodal disturbances that are prevalent in real-world settings. 
Here we propose a novel Mixture of Gaussian Processes based Distributionally Robust MPC (MoGP-DR-MPC) framework for linear time-invariant systems subject to potentially multimodal state-dependent disturbances. 
This framework utilizes MoGP to automatically determine the number of modes from disturbance data.
Using the mean and variance information provided by each mode-specific predictive distribution, it constructs a data-driven state-dependent ambiguity set, which allows for flexible and fine-grained disturbance modeling.
Based on this ambiguity set, we impose Distributionally Robust Conditional Value-at-Risk (DR-CVaR) constraints to effectively achieve distributional robustness against errors in the predictive distributions. To address the computational challenge posed by these constraints in the resulting MPC problem, we equivalently reformulate the DR-CVaR constraints into tractable second-order cone constraints.
Furthermore, we provide theoretical guarantees on the recursive feasibility and stability of the proposed framework.
The enhanced control performance of MoGP-DR-MPC is validated through both numerical experiments and simulations on a quadrotor system, demonstrating notable reductions in closed-loop cost by 17\% and 4\% respectively compared against Gaussian process based MPC.

\end{abstract}


\section{Introduction}
\label{sec:introduction}
Multimodal state-dependent disturbances are pervasive in real-world applications, posing significant challenges to control synthesis because of their inherent complexity and variability. 
Gaussian Process based Model Predictive Control (GP-MPC) offers an effective approach for controlling systems subject to state-dependent disturbances and its control performance largely depends on the accuracy of the GP model.
However, due to the stationary kernels employed in GP, the performance of GP-MPC deteriorates remarkably under multimodal disturbances. Such performance decline underscores the pressing research need for new MPC strategies to tackle the challenges introduced by the heterogeneous modalities. 

In the face of unknown disturbances, different strategies have been proposed for GP-MPC.
Among them, a robust strategy is to enforce constraint satisfaction under the worst-case scenario to ensure complete safety \cite{berkenkamp2015safe}. For instance, References \cite{wang2015robust} and \cite{rose2023learning} used confidence interval bounds and worst-case error bounds, respectively, to tighten constraints.
However, such robust methods tend to be overly conservative. An alternative is to integrate GP with stochastic MPC, which employs chance constraints to substantially mitigate conservatism \cite{bradford2020stochastic}.
To expedite the solution process, Reference \cite{c2} used probabilistic reachable sets to reformulate chance constraints, while \cite{capone2024online} derived a deterministic solution method through online GP-based binary regression.
When the underlying distribution is not perfectly known, stochastic GP-MPC methods fall short of guaranteeing the probabilistic constraint satisfaction. To address this issue, Distributionally Robust MPC (DR-MPC) has been proposed and gained tremendous popularity, ensuring robustness over a so-called ambiguity set\cite{li2023distributionally}\cite{kim2023distributional}. 
DR-MPC can also integrate data-driven techniques to extract moment information or approximate distributions, facilitating the construction of the ambiguity set
\cite{mark2022recursively}\cite{pan2023distributionally}\cite{micheli2022data}.
Notably, these DR-MPC approaches typically assume that disturbances are independent of system states, overlooking state-dependent scenarios.
Drawing upon the capability of DR-MPC to tackle distributional uncertainties, Reference \cite{hakobyan2020learning} introduced Distributionally Robust Conditional Value-at-Risk (DR-CVaR) constraints into GP-MPC, enhancing safety even when the true distribution deviates from the estimated one. 
However, according to the existing literature, a critical research gap remains: most existing GP-MPC research primarily focuses on unimodal disturbances, thus enormously limiting their applicability in real-world settings where multimodal state-dependent disturbances are prevalent.


To fill this challenging gap, we propose a Mixture of Gaussian Processes-based Distributionally Robust MPC (MoGP-DR-MPC) framework. 
To immunize against errors in the predictive distributions, an ambiguity set is constructed based on the structure of modalities and mode-specific predictive distributions.
Within this framework, DR-CVaR constraints are employed in the first predicted system state, while subsequent constraints are tightened via an invariant set \cite{c3}. 
We equivalently reformulate the DR-CVaR constraints as second-order cone constraints to facilitate computation for control actions\cite{c5}, and theoretically provide guarantees for the proposed framework. Its effectiveness and superiority are substantiated through two case studies. The major contributions are summarized as follows:
\begin{itemize}
    \item We develop a novel DR-MPC framework that delicately integrates MoGP with MPC to enhance control performance in the presence of state-dependent and potentially multimodal disturbances, while ensuring recursive feasibility and stability.
    \item We propose an innovative data-driven state-dependent ambiguity set, which closely respects the inherent multimodalities and efficaciously enables the incorporation of fined-grained moment information.
\end{itemize}

\section{PROBLEM FORMULATION}
We consider the following stochastic linear discrete time-invariant system. 
\begin{equation}
    x_{k+1}=Ax_k+Bu_k+w(x_k),\label{dynamics}
\end{equation}
where $x_k\in\mathbb{R}^n,u_k\in\mathbb{R}^m$ and $w(x_k)\in\mathbb{R}^n$ denote the system state, control input and disturbance at time $k$, respectively.
In particular, unknown disturbance $w(\cdot)$ is state-dependent, possibly multimodal and bounded by support $\mathbb{W}$, which is represented by a box constraint set $[\Tilde{w}_1,\Tilde{w}_2]\subseteq\mathbb{R}^n$. 
Assume system matrix $(A,B)$ is known and stabilizable.
The system is subject to the probabilistic state constraints and hard input constraints defined as follows.
\begin{subequations}
\setlength{\abovedisplayskip}{1pt} 
\setlength{\belowdisplayskip}{3pt}
	\begin{align}
  \Pr_{x_k\sim\mathbb{P}}
  (x_k\in&\mathbb{X})\geq 1-\varepsilon,\label{state_cons}\\
    u_k&\in\mathbb{U}, \label{input_cons}
	\end{align}
\end{subequations}
where $\mathbb{X}=[\Tilde{x}_1,\Tilde{x}_2]\subseteq \mathbb{R}^n$ and $\mathbb{U}=\{u\mid H_u u\leq h_u\}$. These sets are compact and contain the origin within their interiors. The notation
$\Pr$ represents the probability and $\varepsilon$ is the prescribed tolerance level that curbs the probability of constraint violation. 

Since the underlying distribution $\mathbb{P}$ is not perfectly known, we propose using a DR-CVaR version of Constraint (\ref{state_cons}), which can be defined as
\begin{equation}
\setlength{\abovedisplayskip}{2pt} 
\setlength{\belowdisplayskip}{2pt}
    \begin{gathered}
            \sup_{\mathbb{Q}\in Q}\mathbb{Q}\text{-CVaR}_{\varepsilon}\left\{(-1)^{i}\left(x_k^{(j)}-\Tilde{x}_i^{(j)}\right)\right\}\leq0 ,\\
    i\in\{1,2\},j\in\{1,\ldots,n\}        
    \label{DR-CVaR}
\end{gathered}
\end{equation}
where the superscript $(\cdot) ^{(j)}$ represents the $j$-th element of the corresponding vector throughout this paper, and $Q$ is an ambiguity set derived from historical data using learning techniques presented in the following section. Although Constraint (\ref{DR-CVaR}) may appear more complex, it admits an exact and convex reformulation that is more tractable in computation than the original chance constraint.

\section{INFINITE MIXTURE of GPs}

To accurately model the state-dependent, potentially multimodal disturbance $w(\cdot)$ and derive the corresponding ambiguity set, we adopt an MoGP approach inspired by the Mixture of experts (MoE) framework \cite{c4}. This method can naturally incorporate the multimodal structure into the model of $w(\cdot)$.

In MoGP, data points are assigned to local experts via a gating network, with each expert then using the assigned data to train its local GP model. Based on this strategy, the observation likelihood for the MoGP model is given by
$$
\Pr(\mathbf{y}|\mathbf{z},\theta)=\prod_i\sum_j \Pr(y_i|c_i=j,z_i,\theta_j)\Pr(c_i=j|z_i,\phi),
$$
where $\mathbf{z}$ represents the training inputs, $\mathbf{y}$ denotes the outputs, $\theta_j$ denotes the parameters of expert $j$, $\phi$ refers to the parameters of the gating network, and $c_i$ represents the indicator variables specifying the corresponding experts. The number of experts is automatically determined from data through maximizing the likelihood.

\subsection{Gating Network Based on Dirichlet Process} 
To accommodate the state-dependent property, we tailor the Dirichlet process to serve as an input-dependent gating network. In the standard Dirichlet process with the concentration parameter $\alpha$, the conditional probability of the indicator variable $c_i$ is given by
\begin{equation}
\setlength{\abovedisplayskip}{2pt} 
\setlength{\belowdisplayskip}{2pt}
    \begin{cases} p(c_i=j|\mathbf{c}_{-i},\alpha)&=\frac{n_{-i,j}}{n-1+\alpha},\;(n_{-i,j}>0)\\
    p(c_i\neq c_{i^{\prime}}\text{ for all }i^{\prime}\neq i|\mathbf{c}_{-i},\alpha)&=\frac{\alpha}{n-1+\alpha},\;(n_{-i,j}= 0)\end{cases} \label{DP}
\end{equation}
where $n_{-i,j}=\sum_{i'\neq i}\delta(c_{i'},j)$ represents the number of data points in expert $j$ excluding observation $i$, and $n$ is the total number of data points. 
To incorporate input dependence into the gating network, we redefine $n_{-i,j}$ as follows.
\begin{equation}
\setlength{\abovedisplayskip}{2pt} 
\setlength{\belowdisplayskip}{2pt}
    n_{-i,j}=(n-1)\frac{\sum_{i^{\prime}\neq i}K_\phi(z_i,z_{i^{\prime}})\delta(c_{i^{\prime}},j)}{\sum_{i^{\prime}\neq i}K_\phi(z_i,z_{i^{\prime}})},\label{nij}
\end{equation}
where $K_{\phi}$ is a kernel function parametrized by $\phi$ and $\delta$ is the Dirac delta function. Consequently, given a new data point, the conditional distribution of its indicator variable can be obtained by substituting (\ref{nij}) into (\ref{DP}).

\subsection{Local Gaussian Process}
For a local GP expert $j$, let $\mathbf{z}=\{z_1,\ldots z_N\}$ represent the training inputs and  $\mathbf{y}=\{y_1,\ldots,y_N\}$ denote the corresponding noisy outputs, where $y_i = w(z_i)+\epsilon_i$. The observation noise $\epsilon_i$ is assumed to be i.i.d, following a normal distribution with mean $0$ and variance $\Sigma_{on}=\text{diag}\{\sigma_{1}^2,\ldots,\sigma_n^2\}$. 
Each output dimension is modeled independently. For dimension $j$, we define a GP prior with kernel $k(\cdot,\cdot)$ and mean function $m(\cdot)$. 
Conditioning on the data $\mathbf{z}$ and $\mathbf{y}^{(j)}=\{y_1^{(j)},\ldots,y_N^{(j)}\}$, 
the predictive posterior distribution for the output at a new data point $z^*$ is Gaussian with the following mean and variance.
\begin{equation}
\setlength{\abovedisplayskip}{2.5pt} 
\setlength{\belowdisplayskip}{2.5pt}
	\begin{aligned}
	\mu_j(z^*)&=k\left({z^*,\mathbf{z}}\right)K_{\sigma_j}^{-1}\mathbf{y}^{(j)} \\
	\Sigma_j(z^*)&=k\left(z^*,z^*\right)-k\left({z^*,\mathbf{z}}\right)K_{\sigma_j}^{-1}k\left({\mathbf{z},z^*} \right)
	\end{aligned}
\end{equation}
where $K_{\sigma_j}=k\left({\mathbf{z},\mathbf{z}}\right)+\sigma_j^2 I$.

Based on the above analysis, the predictive distribution of $w(z^*)$ in dimension $a$ generated by MoGP can be expressed by combining the conditional distributions of the indicator variables from the gating network with the results from local GP regression, as shown in (\ref{gmm}).
\begin{equation}
\setlength{\abovedisplayskip}{3pt} 
\setlength{\belowdisplayskip}{3pt}
    \Pr(y^*\mid \mathcal{D}) = \sum_{j=1}^M \gamma_{a,j}(z^*)\mathcal{N}\left(y^*\mid \mu_{a,j}(z^*), \Sigma_{a,j}(z^*)\right),\label{gmm}
\end{equation}
where $M$ is the number of experts, $\mathcal{D}$ represents the training dataset, and $\gamma_{a,j}$ denotes the weight of the corresponding Gaussian component obtained by Dirichlet process. 
This formulation effectively captures the potential multimodalities in the latent function. 
However, errors between the predictive distribution and the true underlying distribution still exists. To mitigate this uncertainty, we construct an ambiguity set in the following section.
 \vspace{0.2em}

\section{CONTROLLER DESIGN}
With the MoGP model for $w(\cdot)$, we design a controller tailored for system (\ref{dynamics}) in this section. 
To decouple the impact of disturbances, we decompose the system state $x_k$ into nominal component $s_k$ and error component $e_k$.
 The feedback control law is designed as $u_k = Ke_k + v_k$, where $v_k$ is the nominal input and $K$ is a feedback gain matrix computed offline, satisfying that $A+BK$ is stable. Employing this control law, the corresponding nominal and error dynamics are given below \cite{c3}.
\begin{subequations}
\setlength{\abovedisplayskip}{1pt} 
\setlength{\belowdisplayskip}{3pt}
	\begin{align}
	s_{k+1} &=As_k+Bv_k, \label{nominal}\\
	e_{k+1} &=(A+BK)e_k+w(x_k), \label{error}
	\end{align}
\end{subequations}


\subsection{Ambiguity Set Construction}
To formulate the DR-CVaR constraints on system states, we construct a state-dependent ambiguity set that accounts for potential multimodalities by leveraging the MoGP model for $w(\cdot)$ in each dimension.
Particularly, given a specific system state $x^*$ at time $t$, suppose that the predictive distribution for $w(x^*)$ in dimension $\tau$, which consists of $M_\tau$ state-dependent Gaussian components, can be written as
\begin{equation}
\setlength{\abovedisplayskip}{1pt} 
\setlength{\belowdisplayskip}{3pt}
    \sum_{j=1}^{M_{\tau}} \gamma_{\tau,t,j}\mathcal{N}\left(w\mid \mu_{\tau,j}(x^*), \Sigma_{\tau,j}(x^*) \right).\label{GMM}
\end{equation}

Considering the structure of modalities and independence between each mode, we establish local ambiguity sets using the state-dependent mean and variance information from the Gaussian components in (\ref{GMM}). The overall ambiguity set $\mathcal{W}_{\tau}(x^*)$ is then constructed as a weighted Minkowski sum of these local ambiguity sets. Explicitly, each local state-dependent ambiguity set $\mathcal{W}_{\tau,j}$ is defined below.
\begin{equation}
\begin{aligned}
       &\mathcal{W}_{\tau,j}\left(\mathbb{W},x^*\right) \triangleq \\
       &\left\{\rho\in\mathcal{M}_{+} \left|  \begin{aligned}
&\int_{\mathbb{W}}\rho\left(d\xi\right)=1\\
&\int_{\mathbb{W}}\xi\cdot\rho\left(d\xi\right)=\mu_{\tau,j}(x^*)\\
&\int_{\mathbb{W}}\left(\xi-\mu_{\tau,j}(x^*)\right)^2\cdot\rho\left(d\xi\right)\leq\Sigma_{\tau,j}(x^*)
    \end{aligned}\right.\right\}
\end{aligned} \label{local_as}
\end{equation}
where $\mathcal{M}_+$ represents the set of positive Borel measure on $\mathbb{R}$, and $\rho$ is a positive measure. The proposed ambiguity set $\mathcal{W}_\tau(x^*)$ is defined as 
\begin{equation}
    \sum_{j=1}^{M_{\tau}} \gamma_{\tau,t,j}\mathcal{W}_{\tau,j}(\mathbb{W},x^*). \label{mixed ambiguity set}
\end{equation} 

This construction closely aligns with the potentially multimodal structure, allowing for a more flexible and intricate representation of state-dependent disturbances. 

\subsection{Constraint Formulation}
Define $x_{k|t}$ as the $k$ step ahead predicted system state based on the current state $x_t$, while $x_{0|t}=x_t$ is the given initial state. To ensure recursive feasibility,  we adopt a hybrid constraint tightening strategy that combines distributionally robust and worst-case tightening methods.

For $x_{1|t}$, the corresponding DR-CVaR constraint can be expressed as follows \cite{van2015distributionally}.
\begin{equation}
\setlength{\abovedisplayskip}{1pt} 
\setlength{\belowdisplayskip}{3pt}
\begin{gathered}
        \sup_{\mathbb{Q}\in Q_{\tau}(x_{1|t})}\mathbb{Q}\text{-CVaR}_{\varepsilon}\left\{(-1)^{i}\left(x_{1|t}^{(\tau)}-\Tilde{x}_i^{(\tau)}\right)\right\}\leq0 ,\\
    i\in\{1,2\},\tau\in\{1,\ldots,n\}\label{first}
\end{gathered}
\end{equation}

Noting that $x_{1|t} = (A+BK)x_{0|t}-BKs_{0|t} + Bv_{0|t}+w(x_{0|t})$, since the state-dependent term $w(x_{0|t})$ is the only source of uncertainty, the ambiguity set $Q_{\tau}(x_{1|t})$ is essentially a shifted version of the ambiguity set for $w(x_{0|t})^{(\tau)}$, denoted as $\mathcal{W}_\tau(x_{0|t})$. Hence, it can be defined in the same manner as (\ref{mixed ambiguity set}). Based on (\ref{mixed ambiguity set}), we consequently prove that Constraint (\ref{first}) can be equivalently reformulated as second-order cone constraints in the subsequent section.

Given the support set $\mathbb{W}$, we define $\mathcal{Z}$ as the minimal Robust Positively Invariant (mRPI) set for the error dynamics in (\ref{error}).
The remaining system constraints are then guaranteed by the tightened state constraints $s_{k|t} \in \mathbb{X} \ominus \mathcal{Z}, k \in \{2, \ldots, N-1\}$ and input constraints $v_{k|t} \in \mathbb{U} \ominus K\mathcal{Z}, k \in \{0, \ldots, N-1\}$ \cite{mayne2005robust}. Additionally, the constraint $x_t - s_{0|t} \in \mathcal{Z}$ is required to be satisfied. Finally, we impose the terminal constraint $s_{N|t} \in \mathcal{X}_f$, where $\mathcal{X}_f$ is defined in (\ref{terminal}).
\begin{equation}
    \mathcal{X}_f\triangleq\left\{s_{N|k}\in\mathbb{R}^n:
    \begin{gathered}
        s_{N+i|k}\in\mathbb{X}\ominus\mathcal{Z}, \\
        Ks_{N+i|k}\in\mathbb{U}\ominus K\mathcal{Z},
        \forall i\in\mathbb{N}
    \end{gathered}\right\}\label{terminal}
\end{equation}

Based on the above analysis, the formulation of MoGP-DR-MPC at time $t$, denoted as $\mathscr{P}(x_t)$, is written as
\begin{subequations}
	\begin{align}
\min_{s_{0|t},\mathbf{v}_t} &\sum_{k=0}^{N-1}\left(  \left\| s_{k|t} \right\|_{Q}^{2}+\left\| v_{k|t} \right\|_{R}^{2}\right)+\left\| s_{N|t} \right\|_{P}^{2} \notag \\
\text{s.t.}\quad &x_{0|t}=x_t,x_{0|t}-s_{0|t}\in\mathcal{Z} \label{cons_1} \\
& s_{k+1|t}=As_{k|t}+Bv_{k|t},\;k=0,\dots,N-1 \label{cons_2}\\
&\begin{gathered}
        \sup_{\mathbb{Q}\in Q_{\tau}(x_{1|t})}\mathbb{Q}\text{-CVaR}_{\varepsilon}\left\{(-1)^{i}\left(x_{1|t}^{(\tau)}-\Tilde{x}_i^{(\tau)}\right)\right\}\leq0\\
    i\in\{1,2\},\tau\in\{1,\ldots,n\}
\end{gathered}\label{cons_3}\\
&s_{k|t}\in \mathbb{X}\ominus \mathcal{Z},\;k=2,\dots ,N-1 \label{cons_4}\\  
&v_{k|t}\in\mathbb{U}\ominus K\mathcal{Z},\;k=0,\dots,N-1\label{cons_5}\\
&s_{N|t}\in\mathcal{X}_{f} \label{cons_6}
\end{align} 
\label{problem}\end{subequations}
where $\mathbf{v}_t = [v_{0|t} \ldots v_{N-1|t}]^\top$ and $P,Q,R$ are positive definite weight matrices, with $P$ determined through the Riccati equation. Obviously, Problem $\mathscr{P}(x_t)$ is nonconvex and cannot be solved directly due to the presence of constraint (\ref{cons_3}). We then derive its equivalent tractable form in the next section.

\subsection{Tractable MPC Formulation}
In this section, we demonstrate how to transform the critical and intractable DR-CVaR constraint (\ref{cons_3}) into a set of second-order cone constraints using the following theorem.

\begin{theorem}\label{cons_reformulation}
    For any $\tau\in\{1,\ldots,n\}$ and $i\in\{1,2\}$, the DR-CVaR constraint (\ref{cons_3}) is satisfied if and only if 
    \begin{equation}
    \begin{aligned}
        &(-1)^{i} \left[s_{1|t}^{(\tau)}+\left((A+BK)\cdot(x_t-s_{0|t})\right)^{(\tau)}\right]\\
            &\leq (-1)^{i}\Tilde{x}_i^{(\tau)} -\eta_{i,t}^{(\tau)},           
    \end{aligned} \label{ith inequality}
    \end{equation}
    where $\eta_{i,t}^{(\tau)}$ is obtained by the following second-order cone program.
    \begin{equation}
        \begin{aligned}
\min\; & \eta \\
\text { s.t. }\; & 
\varepsilon\beta+\sum_{j=1}^{M_{\tau}(x_{0|t})} \gamma_{\tau,t,j}\left\{
\begin{gathered}
    t_{j}+\mu_{\tau,j}(x_{0|t})\omega_{j}+ \\
    \left[\Sigma_{\tau,j}(x_{0|t})+\mu^2_{\tau,j}(x_{0|t})\right] \Omega_{j}
\end{gathered}
\right\} \leq 0,
\\
& \left\|\begin{array}{c}
\omega_{j}+\varphi_{1,j}-\varphi_{2,j} \\ 
\Omega_j - t_{j} +\Tilde{w}_2^{(\tau)}\varphi_{1,j}-\Tilde{w}_1^{(\tau)}\varphi_{2,j}
\end{array}\right\|_2 \leq \Gamma_{1,j}
,  \\
\\
&\left\|\begin{array}{c}
\omega_{j}-(-1)^i+\phi_{1,j}-\phi_{2,j} \\
\Omega_j- t_{j}-\beta-\eta+\Tilde{w}_2^{(\tau)}\phi_{1,j}-\Tilde{w}_1^{(\tau)}\phi_{2,j}
\end{array}\right\|_2 \leq \Gamma_{2,j},  
\\
&\Omega_{j} \geq 0, \varphi_{1,j} \geq 0,\varphi_{2,j} \geq 0, \phi_{1,j} \geq 0,\phi_{2,j} \geq 0,\\
&\forall j\in\{1,\ldots,M_{\tau}(x_{0|t})\}.
\end{aligned} \label{theorem1's optimization probelm}
    \end{equation}
    where $$
    \begin{aligned}
        \Gamma_{1,j} &= \Omega_j + t_{j} -\Tilde{w}_2^{(\tau)}\varphi_{1,j} + \Tilde{w}_1^{(\tau)}\varphi_{2,j}, \\
    \Gamma_{2,j} &= \Omega_j+ t_{j}+\beta+\eta-\Tilde{w}_2^{(\tau)}\phi_{1,j}+\Tilde{w}_1^{(\tau)}\phi_{2,j}.
    \end{aligned}
    $$
    
    Stacking the $2n$ inequalities defined as (\ref{ith inequality}), the DR-CVaR constraints (\ref{cons_3}) are equivalent to 
    $s_{1|t}+(A+BK)(x_t-s_{0|t})\in \left[\Tilde{x}_1+\eta_{1,t}, \Tilde{x}_2-\eta_{2,t}\right]$.
\end{theorem}
\begin{proof}
    According to the definition of CVaR, the left side of Constraint (\ref{cons_3}) can be rewritten as follows.
    \begin{align}
    \setlength{\abovedisplayskip}{1pt} 
\setlength{\belowdisplayskip}{3pt}
&\sup_{\mathbb{Q}\in Q_{\tau}(x_{1|t})}\mathbb{Q}\text{-CVaR}_{\varepsilon}\left\{(-1)^{i}\left(x_{1|t}^{(\tau)}-\Tilde{x}_i^{(\tau)}\right)\right\} \label{CVaR def} \\
=&\inf_{\beta\in\mathbb{R}}\left\{\beta+\frac{1}{\varepsilon}\sup_{\mathbb{P}\in\mathcal{W}_{\tau}(x_{0|t})}\mathbb{E}_\mathbb{P}\left(
\begin{gathered}
    (-1)^i\left(z_t^{(\tau)}-\Tilde{x}_i^{(\tau)}\right)\\
    +(-1)^i w(x_{0|t})^{(\tau)}-\beta
\end{gathered}
\right)^+\right\} \notag
\end{align}
    where $z_t=(A+BK)(x_t-s_{0|t})+ s_{1|t}$. To simplify the notation, we denote $(-1)^i\left(z_t^{(\tau)}-\Tilde{x}_i^{(\tau)}\right)$ as $C(i,\tau)$. Given the initial state $x_{0|t}=x_t$, the ambiguity set $\mathcal{W}_{\tau}(x_{0|t})$ is defined similarly as (\ref{mixed ambiguity set}). Then the worst-case expectation term $\sup_{\mathbb{P}\in\mathcal{W}_{\tau}(x_{0|t})}\mathbb{E}_\mathbb{P}\left(C(i,\tau)+(-1)^i w(x_{0|t})^{(\tau)}-\beta\right)^+$ can be reformulated as follows \cite{c5} \cite{hanasusanto2015distributionally}.

        \begin{align}\sup_{\boldsymbol{\rho}} &\sum_{j=1}^{M_{\tau}(x_{0|t})}\gamma_{\tau,t,j}\int_{\mathcal{W}_{\tau}(x_{0|t})}\left\{C(i,\tau)+(-1)^i\xi-\beta\right\}^+ \rho_j\left(d\xi\right) \notag\\ 
         \text{s.t.} 
         &\begin{aligned}
             &\int_{\mathcal{W}_{\tau}(x_{0|t})}\rho_j\left(d\xi\right)=1\\
       &\int_{\mathcal{W}_{\tau}(x_{0|t})}\xi \rho_j\left(d\xi\right)=\mu_{\tau,j}(x_{0|t})\\&\int_{\mathcal{W}_{\tau}(x_{0|t})}(\xi-\mu_{\tau,j}(x_{0|t}))^2\rho_j\left(d\xi\right)\leq\Sigma_{\tau,j}(x_{0|t})
         \end{aligned}\label{mixed-1}
        \end{align}

     Obviously, Problem (\ref{mixed-1}) cannot be solved directly, so we take its dual, which is formulated below.
    \begin{equation}
        \begin{aligned}
\min_{t_{j},\omega_{j},\Omega_{j}}\; & \sum_{j=1}^{M_{\tau}(x_{0|t})}\gamma_{\tau,t,j}\left\{
\begin{gathered}
    t_{j}+\mu_{\tau,j}(x_{0|t})\omega_{j}\\
    +\left[\Sigma_{\tau,j}(x_{0|t})+\mu_{\tau,j}^2(x_{0|t})\right]\cdot \Omega_{j}
\end{gathered}
\right\} \\
\text{s.t.}\; & t_{j}+\xi\omega_{j}+\xi^2\Omega_{j}\geq0,   \\
&-\left(C(i,\tau)+(-1)^i\xi-\beta\right)+t_{j}
+\xi\omega_{j}+\xi^2\Omega_{j}\geq 0, \\
&\Omega_{ij}\geq0, \Tilde{w}_1^{(\tau)}\leq\xi\leq \Tilde{w}_2^{(\tau)} ,\\
&\forall 1\leq j\leq M_{\tau}(x_{0|t}).
\end{aligned} \label{sup-dual}
    \end{equation}
    
    The semi-infinite constraints in (\ref{sup-dual}) can be transformed into linear matrix inequalities by duality theory and Schur complements \cite{boyd2004convex}. For brevity, the exact form is omitted. Since the resulting matrices are two-dimensional, they can be further rewritten as second-order cone constraints \cite{10287922}. Thus, Constraint (\ref{cons_3}) is equivalent to the following set of constraints.
    \begin{equation}
        \begin{aligned}
&\varepsilon\beta+\sum_{j=1}^{M_{\tau}(x_{0|t})} \gamma_{\tau,t,j}\left\{
\begin{gathered}
    t_{j}+\mu_{\tau,j}(x_{0|t})\omega_{j}+ \\
    \left[\Sigma_{\tau,j}(x_{0|t})+\mu^2_{\tau,j}(x_{0|t})\right] \Omega_{j}
\end{gathered}
\right\} \leq 0, \\
&\left\|\begin{array}{c}
\omega_{j}+\varphi_{1,j}-\varphi_{2,j} \\ 
\Omega_j - t_{j} +\Tilde{w}_2^{(\tau)}\varphi_{1,j}-\Tilde{w}_1^{(\tau)}\varphi_{2,j}
\end{array}\right\|_2 \leq \Tilde{\Gamma}_{1,j}
,  \\
\\
&\left\|\begin{array}{c}
\omega_{j}-(-1)^i+\phi_{1,j}-\phi_{2,j} \\
\Omega_j- t_{j}-\beta+C(i,\tau)+\Tilde{w}_2^{(\tau)}\phi_{1,j}-\Tilde{w}_1^{(\tau)}\phi_{2,j}
\end{array}\right\|_2 \leq \Tilde{\Gamma}_{2,j} 
\\
&\Omega_{j} \geq 0, \varphi_{1,j} \geq 0,\varphi_{2,j} \geq 0, \phi_{1,j} \geq 0,\phi_{2,j} \geq 0,\\
&\forall 1\leq j\leq M_{\tau}(x_{0|t}) \label{constraints_set}
\end{aligned}
    \end{equation}
    where $$\begin{aligned}
        \Tilde{\Gamma}_{1,j} &= \Omega_j + t_{j} -\Tilde{w}_2^{(\tau)}\varphi_{1,j} + \Tilde{w}_1^{(\tau)}\varphi_{2,j},\\
    \Tilde{\Gamma}_{2,j} &= \Omega_j+ t_{j}+\beta-C(i,\tau)-\Tilde{w}_2^{(\tau)}\phi_{1,j}+\Tilde{w}_1^{(\tau)}\phi_{2,j}.
    \end{aligned}
    $$
    
    Substituting (\ref{constraints_set}) into (\ref{CVaR def}) gives (\ref{theorem1's optimization probelm}),  which completes the proof.
\end{proof}

Theorem \ref{cons_reformulation} provides the convex reformulation alternatives of Constraint (\ref{cons_3}). It follows that the optimization problem $\mathscr{P}(x_t)$ can be reformulated into the convex and tractable form below, which is denoted as $\mathscr{P}_{MoDR}(x_t)$.
\begin{subequations}
\setlength{\abovedisplayskip}{2pt} 
\setlength{\belowdisplayskip}{2pt}
	\begin{align}
\min_{s_{0|t},\mathbf{v}_t} &\sum_{k=0}^{N-1}\left(  \left\| s_{k|t} \right\|_{Q}^{2}+\left\| v_{k|t} \right\|_{R}^{2}\right)+\left\| s_{N|t} \right\|_{P}^{2} \notag \\
\text{s.t.}\quad &x_{0|t}-s_{0|t}\in\mathcal{Z} ,\label{invariant} \\
& s_{k+1|t}=As_{k|t}+Bv_{k|t},\;k=0,\dots,N-1 \label{final_2}\\
&
\begin{aligned}
s_{1|t}+&(A+BK)(x_t-s_{0|t}),\\ &\in \left[\Tilde{x}_1+\eta_{1,t}, \Tilde{x}_2-\eta_{2,t}\right]
\end{aligned}
\label{final_3}\\
&s_{k|t}\in\mathbb{X}\ominus\mathcal{Z},\;k=2,\dots ,N-1 \label{final_4}\\  
&v_{k|t}\ominus \mathbb{U}\ominus K\mathcal{Z},\;k=0,\dots,N-1\label{final_5}\\
&s_{N|t}\in\mathcal{X}_f.\label{final_6}
\end{align} \label{final_problem}
\end{subequations}

If $\mathscr{P}_{MoDR}(x_t)$ is feasible at time $t$, let $s_{0|t}^*$ and $\mathbf{v}_t^*=[v_{0|t}^*,\ldots,v_{N-1|t}^*]^\top$ represent its optimal solution. 
At each time instant, the MoGP-DR-MPC framework solves $\mathscr{P}_{MoDR}(x_t)$ to obtain the optimal input sequence and only the first element $u_t^*=K(x_t-s_{0|t}^*) + v_{0|t}^*$ is executed. This process is repeated to perform receding horizon control.

\section{THEORETICAL PROPERTIES}
In this section, we establish the recursive feasibility and stability of MoGP-DR-MPC through the following theorems.

\begin{theorem}[Recursive feasibility] \label{recursive_feasibilty}
    If problem $\mathscr{P}_{MoDR}(x_t)$ is feasible with state $x_t$, then $\mathscr{P}_{MoDR}(x_{t+1})$ is recursively feasible with $u_t=K(x_t-s_{0|t}^*) + v_{0|t}^*$.
\end{theorem}
\begin{proof}
    Given the optimal sequences $\mathbf{v}_t^*$ and $\mathbf{s}_t^*$, from Constraint (\ref{invariant}), we know that $e_t = x_t-s_{0|t}^*\in\mathcal{Z}$. By the invariant property of $\mathcal{Z}$, we have $e_{t+1} = (A+BK)e_t + w(x_t)\in\mathcal{Z}$ for all possible $w(x_t)\in\mathbb{W}$. Define the solution sequence of problem $\mathscr{P}_{MoDR}(x_{t})$ as $\Lambda_t^* = [s_{0|t}^*, v_{0|t}^*,\ldots,v_{N-1|t}^*]$. We then shift $\Lambda_t^*$ to generate the candidate solution for problem $\mathscr{P}_{MoDR}(x_{t+1})$. In particular, $\Lambda_{t+1}=[s_{1|t}^*, v_{1|t}^*,\ldots, Ks_{N|t}^*]$.

    Since $e_{t+1}\in\mathcal{Z}$, $x_{t+1}-s_{1|t}^*\in \mathcal{Z}$, i.e. Constraint (\ref{invariant}) is satisfied. Given the feasibility of $\Lambda_t^*$, Constraints (\ref{final_2})-(\ref{final_5}) are satisfied by $[s_{1|t}^*,\ldots,(A+BK)s_{N|t}^*]$ and $[v_{1|t}^*,\ldots,v_{N-1|t}^*]$. According to the definition of $\mathcal{X}_f$, $(A+BK)s_{N|t}^*\in\mathcal{X}_f$ and $Ks_{N|t}^*\in \mathbb{U}\ominus K\mathcal{Z}$, so that the terminal constraint (\ref{final_6}) is also satisfied. Therefore, $\Lambda_{t+1}$ is feasible at time $t+1$. 
\end{proof}


\begin{theorem}[Closed-loop stability] \label{closed-loop stability}
    If the initial problem $\mathscr{P}_{MoDR}(x_{0})$ is feasible, 
    the closed-loop system asymptotically converges to a neighborhood of the origin.
\end{theorem}
\begin{proof}
    Consider $\Lambda_t^*$ at time $t$, the optimal objective value is denoted as $J_t^*$.
    Since candidate solution $\Lambda_{t+1}$ is feasible, the corresponding objective value is obtained by
    \begin{equation*}
        \begin{aligned}
        J_{t+1}&=
             \sum_{k=1}^{N-1}\left( \left\| s^{*}_{k|t} \right\|_{Q}^{2}+\left\| v^{*}_{k|t} \right\|_{R}^{2}\right)+\left\| s^{*}_{N|t} \right\|_{Q}^{2}\\
             &+\left\| Ks_{N|t}^{*} \right\|_{R}^{2}+\left\| (A+BK)s_{N|t}^{*} \right\|_{P}^{2}.
        \end{aligned}
    \end{equation*}
    
    Noting that $P$ is the solution of Riccati equation, we obtain the following inequality.
    \begin{equation}
        \begin{aligned}
            J_{t+1}^*&\leq J_{t+1} \\
            &\leq \sum_{k=0}^{N-1}\left(  \left\| s^*_{k|t} \right\|_{Q}^{2}+\left\| v^*_{k|t} \right\|_{R}^{2}\right)+\left\| s^*_{N|t} \right\|_{P}^{2}\\
            &=J_{t}^*
        \end{aligned}
    \end{equation}
    
    Therefore, $J_{t+1}^*-J_{t}^*\leq \left\| s^*_{0|t} \right\|_{Q}^{2}+\left\| v^*_{0|t} \right\|_{R}^{2}$. Adding this from $t=0$, we obtain $J_0^*-J_{\infty}^*\leq \sum_{t=0}^{\infty} \left\| s^*_{0|t} \right\|_{Q}^{2}+\left\| v^*_{0|t} \right\|_{R}^{2}$. Then we have $\lim_{t\to \infty} \left\| s^*_{0|t} \right\|_{Q}^{2}=0$ and $\lim_{t\to \infty} \left\| v^*_{0|t} \right\|_{R}^{2}=0$. This gives the convergence of system states.
\end{proof}

\section{CASE STUDIES}
This section presents two case studies to validate the effectiveness of the proposed MoGP-DR-MPC method. Its performance is compared against two baseline methods, namely GP-based DR-MPC (GP-DR-MPC) and robust tube MPC. 
\subsection{Numerical Experiments}
We first conduct numerical experiments on a stochastic system given in (\ref{numerical_system}). 
\begin{equation}
    x^+ = \begin{bmatrix}
        1 & 1\\ 0 &1 
    \end{bmatrix} x + \begin{bmatrix}
        0.5\\1
    \end{bmatrix} u + w(x) \label{numerical_system}
\end{equation}

Suppose the constraint sets are $\mathbb{X}=[-7,0]\times[-3,2]$ and $\mathbb{U}=[-5,5]$, with a tolerance level of $\varepsilon=0.2$. Disturbance $w(x)$ is supported by $[-0.8,0.8]\times [-0.8,0.8]$ and is defined using standard four-modal function i.e. Franke function \cite{franke1979critical} in each dimension. The control objective is to steer the initial state $[-5,-2]$ to the origin under the disturbance $w(x)$. Define the stage cost as $x^\top Q x + u^\top R u$, where $Q$ is identity matrix $I_2$ and $R=0.1$.
The horizon of MPC is set as $10$ and we perform 50 closed-loop simulations with a simulation horizon of 30 to showcase the improvement in control performance. 

The average closed-loop costs for MoGP-DR-MPC, GP-DR-MPC and robust tube MPC are 66.06, 79.74 and 89.22, respectively.
The simulation results indicate that MoGP-DR-MPC reduces the average closed-loop cost by 17\% compared with GP-DR-MPC and by 24\% compared with robust tube MPC, demonstrating its notable advantage in handling multimodal state-dependent disturbances.


\subsection{Simulations on a Quadrotor System}
In the second case study, we simulate a trajectory planning and control problem for a planar quadrotor system \cite{nakka2018six}. The quadrotor starts at $(10,10)$ and is expected to land at the origin. During its flight, it must remain within the safe region $\mathcal{X}$ with probability $0.8$ under a multimodal state-dependent wind disturbance, which is common due to varying airflow patterns. The linearized dynamics of the quadrotor system are expressed as follows.

\begin{equation}
\begin{gathered}
        s_{k+1|t} = A s_{k|t} + B u_{k|t} + w(s_{k|t}),
    \\
    A=\begin{bmatrix}1&\Delta t&0&0\\0&1&0&0\\0&0&1&\Delta t\\0&0&0&1\end{bmatrix}, B=\frac{1}{m}\begin{bmatrix}\frac{\Delta t^2}{2}&0\\\Delta t&0\\0&\frac{\Delta t^2}{2}\\0&\Delta t\end{bmatrix},
\end{gathered}
\end{equation}
where $\Delta t=1$ is the sampling time and $m$ represents the mass of the system. 
The state vector $s_{t}=[p_x, v_{x}, p_{y}, v_{y}]^\top\in\mathbb{R}^4$ includes the quadrotor's positions and velocities in the $x$ and $y$ directions, while the control input $u_t = [u_x, u_y]^\top\in [-7,7]^2\subseteq \mathbb{R}^2$ applies forces along these axes. The support set of the wind disturbance is represented by $\mathbb{W}=[-0.6,0.6]^4\subseteq \mathbb{R}^4$ and the safe region $\mathcal{X}$ is defined by
\begin{equation*}
    \left\{s:
    \begin{aligned}
        \begin{bmatrix}
        -4 &-4 &-4 &-4
    \end{bmatrix}^\top\leq s\leq   \begin{bmatrix}
        18  &4  &18  &4
    \end{bmatrix}^\top
\end{aligned}
\right\}.
\end{equation*}

In the simulations, the MPC time horizon $N$ is set to five and we evaluate the performance of MoGP-DR-MPC against two baseline approaches. Specifically, we analyze the closed-loop costs under 50 different sequences of disturbance realizations, each with 30 simulation steps. As illustrated in Fig \ref{fig:cost}, MoGP-DR-MPC achieves an average reduction in closed-loop cost of 4\% compared with GP-DR-MPC and 5\% compared with robust tube MPC. This corroborates the superior ability of MoGP-DR-MPC to capture the multimodal characteristics of disturbance data, yielding less conservative control performance.

 Additionally, because the optimal control problem is convex and the MoGP model is trained offline, the computation time for MoGP-DR-MPC is on par with that of conventional GP-MPC. Fig \ref{fig:traj} shows the position trajectories of the quadrotor generated by 30 simulations.

 \begin{figure}[htbp]
    \centering
    \includegraphics[width=0.9\linewidth]{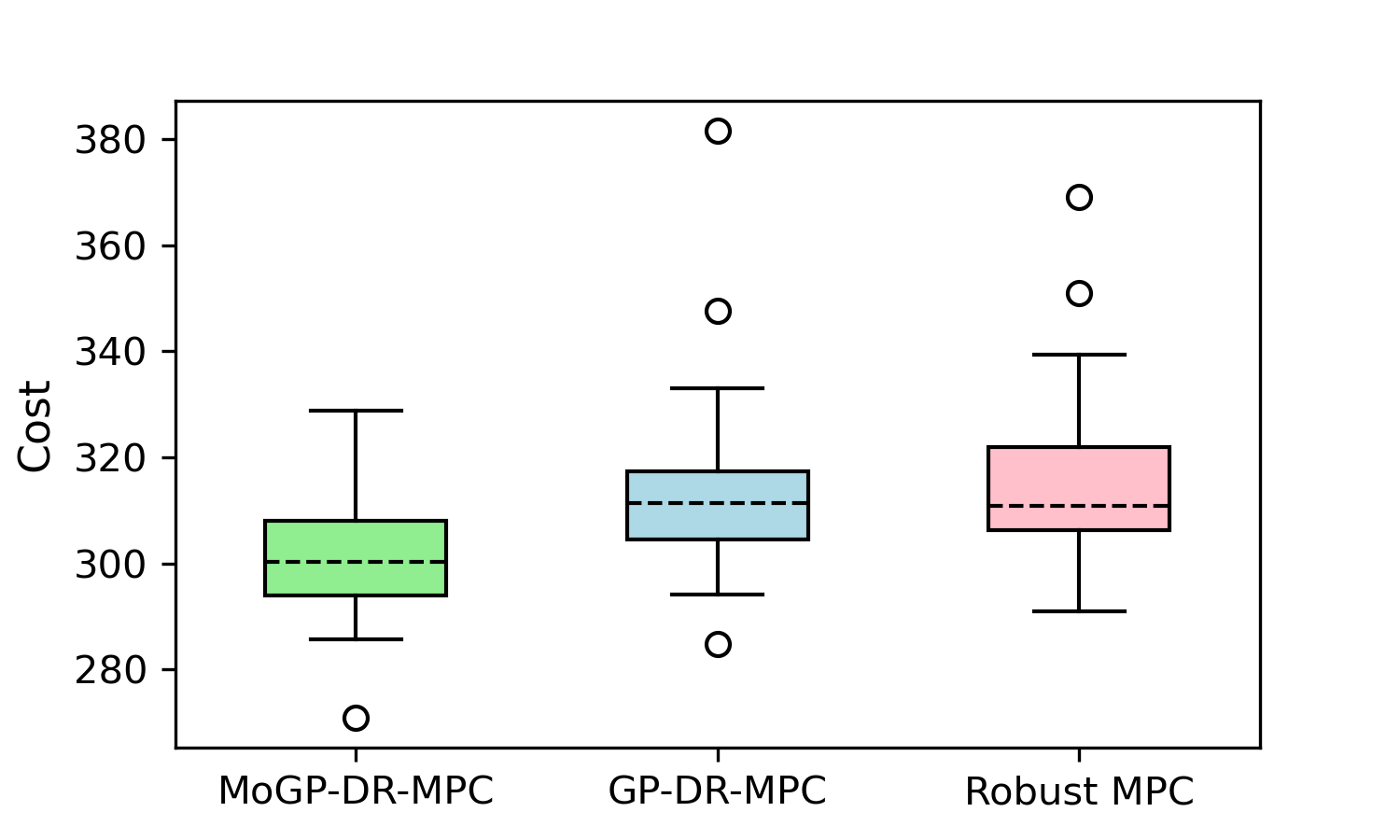}
    \caption{The closed-loop cost of MoGP-DR-MPC and two baseline methods under 50 realizations of disturbance sequences.}
    \label{fig:cost}
\end{figure}
\begin{figure}[htbp]
    \centering
    \includegraphics[width=0.8\linewidth]{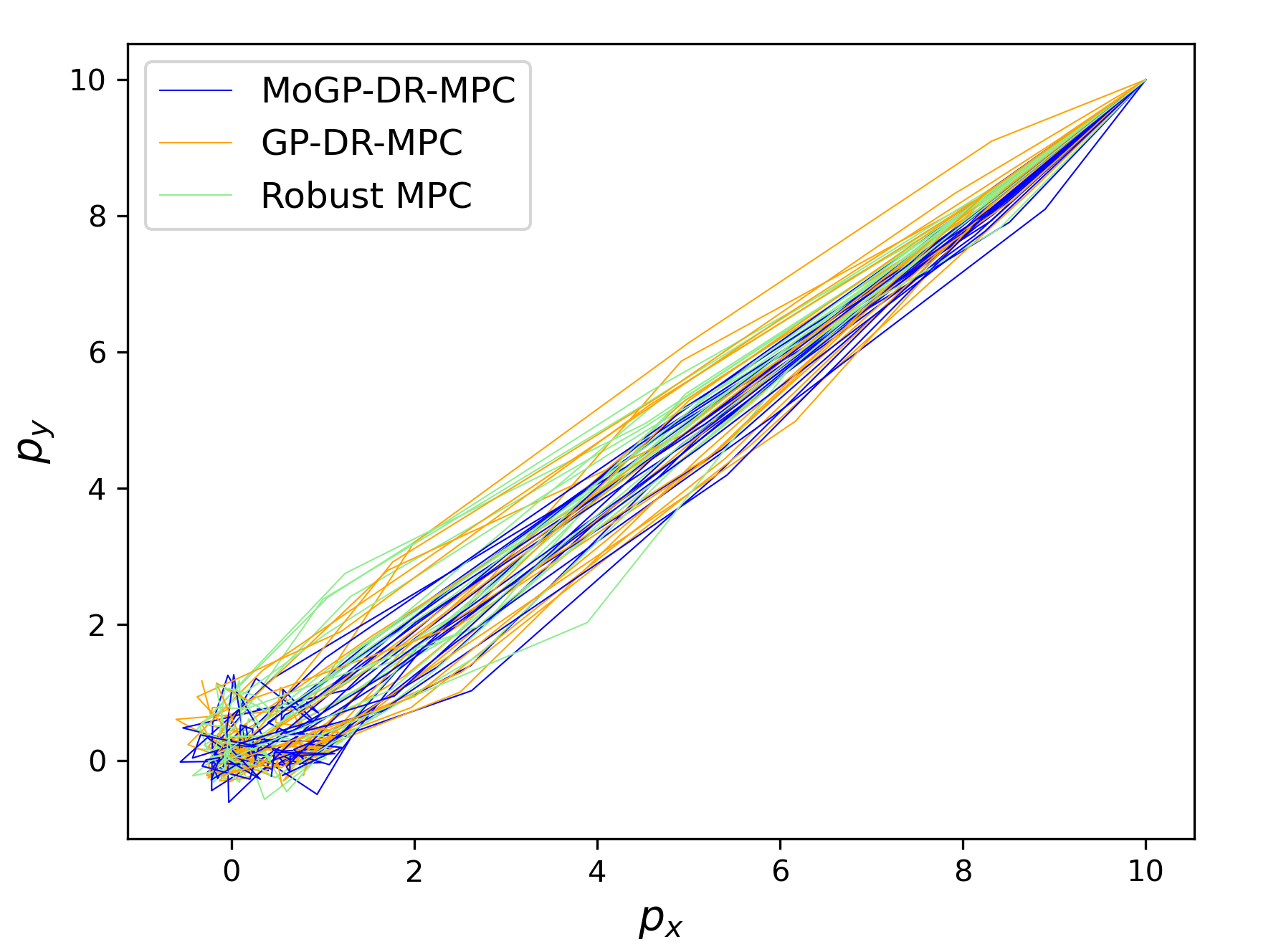}
    \caption{The closed-loop trajectories of the proposed MoGP-DR-MPC strategy and two baseline methods under different disturbance realizations.}
    \label{fig:traj}
\end{figure}
\section{CONCLUSIONS}
In this paper, we presented a novel MoGP-DR-MPC framework that seamlessly integrated MoGP to effectively control systems with state-dependent multimodal disturbances. Based on the MoGP model, we developed a data-driven state-dependent ambiguity set, which closely aligns with the multimodality structure. We reformulated the DR-CVaR constraints as scalable second-order cone constraints, ensuring computational tractability. The recursive feasibility and closed-loop stability of MoGP-DR-MPC were guranteed through invariant sets. Both numerical experiments and simulations on a quadrotor system verified the effectiveness of the proposed method in coping with multimodal disturbances.

\bibliographystyle{IEEEtran}
\bibliography{reference}

\end{document}